\documentclass[12pt]{article}

\usepackage{mathpazo}
\usepackage[colorlinks, linkcolor=blue, citecolor=blue]{hyperref}           
\usepackage{color,mathrsfs}
\usepackage{graphicx,subfigure,amsmath,amssymb,amsfonts,bm,epsfig,epsf,url,dsfont}
\usepackage{amsthm}
\usepackage{tikz}
\usepackage{bbm}      
\usepackage{booktabs}
\usepackage{cases}
\usepackage{fullpage}
\usepackage[small,bf]{caption}
\usepackage[top=1in,bottom=1in,left=1in,right=1in]{geometry}
\usepackage{fancybox}
\usepackage{algorithm}
\usepackage{verbatim}
\usepackage{algorithmic}

\usepackage{mathtools}

\usepackage{scalerel,stackengine}
\stackMath
\newcommand\reallywidehat[1]{%
\savestack{\tmpbox}{\stretchto{%
  \scaleto{%
    \scalerel*[\widthof{\ensuremath{#1}}]{\kern-.6pt\bigwedge\kern-.6pt}%
    {\rule[-\textheight/2]{1ex}{\textheight}}
  }{\textheight}%
}{0.5ex}}%
\stackon[1pt]{#1}{\tmpbox}%
}

\newcommand{\bE}{\mathbb{E}}

\newtheorem{definition}{Definition}

\newtheorem{theorem}{Theorem}

\newtheorem{conjecture}{Conjecture}
\newtheorem{lemma}{Lemma}

\newtheorem{problem}{Problem}

\newtheorem{rmk}{Remark}

\newenvironment{fminipage}%
  {\begin{Sbox}\begin{minipage}}%
  {\end{minipage}\end{Sbox}\fbox{\TheSbox}}

\makeatletter
\newcommand*{\rom}[1]{\expandafter\@slowromancap\romannumeral #1@}
\makeatother

\newcommand{\Ind}{\mathbbm{1}}

\newcommand {\pr} {\mathbb{P}}

\newcommand{\calA}{{\cal A}}

\newcommand{\calC}{{\cal C}}

\newcommand{\calE}{{\cal E}}

\newcommand{\calG}{{\cal G}}
\newcommand{\calH}{{\cal H}}

\newcommand{\calK}{{\cal K}}
\newcommand{\calL}{{\cal L}}
\newcommand{\calM}{{\cal M}}

\newcommand{\calO}{{\cal O}}
\newcommand{\calP}{{\cal P}}
\newcommand{\calQ}{{\cal Q}}
\newcommand{\calR}{{\cal R}}
\newcommand{\calS}{{\cal S}}

\newcommand{\calW}{{\cal W}}
\newcommand{\calX}{{\cal X}}

\newcommand{\be}{\begin{equation}}
\newcommand{\ee}{\end{equation}}
\newcommand{\beqna}{\begin{eqnarray}}
\newcommand{\eeqna}{\end{eqnarray}}

\newcommand{\p}[1]{\left(#1\right)}
\newcommand{\pp}[1]{\left[#1\right]}
\newcommand{\ppp}[1]{\left\{#1\right\}}

\usepackage{xspace}

\newcommand{\s}[1]{\mathsf{#1}}

\usepackage{tikz}
\usetikzlibrary{shapes,arrows}
\usepackage{verbatim}

\tikzstyle{block} = [draw, fill=blue!20, rectangle, 
    minimum height=3em, minimum width=6em]
\tikzstyle{sum} = [draw, fill=blue!20, circle, node distance=1cm]
\tikzstyle{input} = [coordinate]
\tikzstyle{output} = [coordinate]
\tikzstyle{pinstyle} = [pin edge={to-,thin,black}]

\allowdisplaybreaks
\begin{document}

\title{Random Subgraph Detection Using Queries}

\author{Wasim~Huleihel\thanks{W. Huleihel is with the Department of Electrical Engineering-Systems at Tel Aviv university, {T}el {A}viv 6997801, Israel (e-mail:  \texttt{wasimh@tauex.tau.ac.il}). The work of W. Huleihel was supported by the ISRAEL SCIENCE FOUNDATION (grant No. 1734/21).}~~~~~~Arya~Mazumdar\thanks{A.Mazumdar is with the Halıcıoğlu Data Science Institute at University of California, San Diego, USA (email: \texttt{arya@ucsd.edu).} The work of A. Mazumdar was supported by NSF under Award 2217058, and Award 2133484.}~~~~~~Soumyabrata~Pal\thanks{S. Pal is with Google Research, Bengaluru, India (email: \texttt{soumyabrata@google.com}).}}


\maketitle

\begin{abstract}

The planted densest subgraph detection problem refers to the task of testing whether in a given (random) graph there is a subgraph that is unusually dense. Specifically, we observe an undirected and unweighted graph on $n$ vertices. Under the null hypothesis, the graph is a realization of an Erd\H{o}s-R\'{e}nyi graph with edge probability (or, density) $q$. Under the alternative, there is a subgraph on $k$ vertices with edge probability $p>q$. The statistical as well as the computational barriers of this problem are well-understood for a wide range of the edge parameters $p$ and $q$. In this paper, we consider a natural variant of the above problem, where one can only observe a relatively small part of the graph using edge queries. This problem is relevant, for example, when access to the edges (connections) between vertices (individuals) may be scarce due to privacy concerns. 

For this model, we determine the number of queries necessary and sufficient for detecting the presence of the planted subgraph. Specifically, we show that any (possibly randomized) algorithm must make $\s{Q} = \Omega(\frac{n^2}{k^2\chi^4(p||q)}\log^2n)$ non-adaptive queries (on expectation) to the adjacency matrix of the graph to detect the planted subgraph with probability more than $1/2$, where $\chi^2(p||q)$ is the Chi-Square distance. On the other hand, we devise a quasi-polynomial-time algorithm that detects the planted subgraph with high probability by making $\s{Q} = O(\frac{n^2}{k^2\chi^4(p||q)}\log^2n)$ non-adaptive queries. We then propose a polynomial-time algorithm which is able to detect the planted subgraph using $\s{Q} = O(\frac{n^3}{k^3\chi^2(p||q)}\log^3 n)$ queries. We conjecture that in the leftover regime, where $\frac{n^2}{k^2}\ll\s{Q}\ll \frac{n^3}{k^3}$ and $k = \omega(\sqrt{n})$, no polynomial-time algorithms exist. Our results resolve two open questions posed in the past literature, where the special case of detection and recovery of a planted clique (i.e., $p=1$) was considered. 
    
\end{abstract}

\section{Introduction}\label{sec:intro}

In the planted densest subgraph ($\s{PDS}$) formulation of community detection, the task is to detect the presence of a small subgraph of size $k$ planted in an Erd\H{o}s-R\'{e}nyi random graph. This problem has been studied extensively both from the algorithmic and the information-theoretic perspectives
\cite{arias2014community, butucea2013detection, verzelen2015community, chen2016statistical, montanari2015finding, candogan2018finding, hajek2016achieving}. Nonetheless, the best known algorithms exhibit a peculiar phenomenon: there appears to be a statistical-computational gap between the minimum $k$ at which this task can be solved and the minimum $k$ at which it can be solved in polynomial-time. Tight statistical-computational bounds for several parameter regimes of the $\s{PDS}$ were recently established through average-case reductions from the planted clique conjecture \cite{ma2015computational, hajek2015computational, brennan2018reducibility}. The regimes in which these problems are information-theoretically impossible, statistically possible but computational hard, and admit polynomial-time algorithms appear to have a common structure.

Recently, models of clustering and community detection that allow active querying for pairwise similarities have become quite popular. This includes active learning, as well as data labeling by amateurs via crowdsourcing. Clever implementation of an interactive querying framework can improve the accuracy of clustering and help in inferring labels of large amount of data by issuing only a small number of queries. Queries can be easily implemented, e.g., via captcha. Non-expert workers in crowdsourcing platforms are often not able to label the items directly, however, it is reasonable to assume that they can compare items and judge whether they are similar or not. Understanding the query complexity to recover hidden structures is a fundamental theoretical question with several applications, from community detection to entity resolution~\cite{Mazumdar2017ClusteringWN,Mazumdar2017QueryCO}. For example, analyzing query complexity and designing query-based algorithms is relevant to community recovery in social networks, where access to the connections (edges) between individuals may be limited due to privacy concerns; or the network can be very large so only part of the graph can be sampled.

In this paper, we investigate a natural variant of the $\s{PDS}$ problem above, where one can only observe a small part of the graph using non-adaptive edge queries. A precise description of our model as well as our main goals are described next.

\subsection{Model Formulation and Goal}

To present our model, we start by reminding the reader the basic mathematical formulation of the $\s{PDS}$ detection problem. Specifically, let $\calG(n,q)$ denote the Erd\H{o}s-R\'{e}nyi random graph with $n$ vertices, where each pair of vertices is connected independently with probability $q$. Also, let $\calG(n,k,p,q)$ with $p>q$ denote the ensemble of random graphs generated as follows:
\begin{itemize}
    \item Pick $k$ vertices uniformly at random from $[n]\triangleq\{1,2,\ldots,n\}$, and denote the obtained set by $\calK$.
    \item Any two vertices in $\calK$ are connected with probability $p$; all other pairs of vertices are connected with probability $q$.
\end{itemize}
In summary, $\calG(n,k,p,q)$ is the ensemble of graphs of size $n$, where a random subgraph $\calG(k,p)$ is \emph{planted} in an Erd\H{o}s-R\'enyi random graph $\calG(n,q)$; this ensemble is known as the $\s{PDS}$ model. The vertices in $\calK$ form a \emph{community} with higher density than elsewhere. In this paper, we focus on the regime where both edge probabilities $p$ and $q$ are fixed and independent of $n$. The $\s{PDS}$ \emph{detection problem} is defined as follows.
\begin{definition}[$\s{PDS}$ detection problem]\label{def:PDS}
The $\s{PDS}$ detection problem with parameters $(n,k,p,q)$, henceforth denoted by $\s{PDS}(n,k,p,q)$, refers to the problem of distinguishing hypotheses:
\begin{align}
    \calH_0: \s{G}_n\sim\calG(n,q)\quad\quad\s{vs.}\quad\quad\calH_1: \s{G}_n\sim\calG(n,k,p,q).\label{eqn:model}
\end{align}
\end{definition}
The statistical and computational barriers of the problem in Definition~\ref{def:PDS} depend on the parameters $(n,k,p,q)$. Roughly speaking, if the planted subgraph size $k$ decreases, or if the ``distance" between the densities $p$ and $q$ decrease, the distributions under the null and alternative hypotheses become less distinguishable. The statistical limits (i.e., necessary and sufficient conditions) for detecting planted dense subgraphs, without any constraints on the computational complexity, were established in \cite{arias2014community,verzelen2015community}. Interestingly, in the same papers it was observed that state-of-the-art low-complexity algorithms are highly suboptimal. This raised the intriguing question of whether those gaps between the amount of data needed by all computationally efficient algorithms and what is needed for statistically optimal algorithms is inherent. According, quite recently \cite{hajek2015computational,brennan2018reducibility,Brennan2019UniversalityOC}, tight statistical-computational gaps for several parameter regimes of $\s{PDS}$ were established through average-case reductions from the planted clique conjecture (see, Conjecture~\ref{conj:PC} below). 

In this work, we consider a variant of the $\s{PDS}$ detection problem where one can only inspect a small part of the graph by \emph{non-adaptive edge queries}, defined as follows.
\begin{definition}[Oracle/Edge queries]\label{def:oracle}
Consider a graph $\s{G}_n=([n],\calE)$ with $n$ vertices, where $\calE$ denotes the set of edges. An oracle $\calO:[n]\times[n]\to\{0,1\}$, takes as input a pair of vertices $i,j\in[n]\times[n]$, and if $(i,j)\in\calE$, namely, there exists an edges between the chosen vertices, then $\calO(i,j)=1$, otherwise, $\calO(i,j)=0$.
\end{definition}

We consider query mechanisms that evolve dynamically over $\s{Q}$ steps/queries in the following form: in step number $\ell\in[\s{Q}]$, the mechanism chooses a pair of vertices $\s{e}_\ell\triangleq(i_\ell,j_\ell)$ and asks the oracle whether these vertices are connected by an edge or not. Generally speaking, either \emph{adaptive} or \emph{non-adaptive} query mechanisms can be considered. In the former, the chosen $\ell$th pair may depend on the previously chosen pairs $\{\s{e}_i\}_{i<\ell}$, as well as on past responses $\{\calO(\s{e}_i)\}_{i<\ell}$. In non-adaptive mechanisms, on the other hand, all queries must be made upfront. In this paper we focus mainly on non-adaptive mechanisms. The query-$\s{PDS}$ detection problem is defined as follows.
\begin{problem}[Query-$\s{PDS}$ detection problem]\label{def:PDSquery}
Consider the $\s{PDS}$ detection problem in Definition~\ref{def:PDS}.  There is an oracle $\calO$ as defined in Definition~\ref{def:oracle}. Find a set of queries $\mathbb{Q}\subseteq[n]\times[n]$ such that
$\s{Q}=|\mathbb{Q}|$, and from the oracle answers it is possible to solve (as defined below) the detection problem in \eqref{eqn:model}. Henceforth, we denote this detection problem by $\s{QPDS}(n,k,p,q,\s{Q})$.
\end{problem}
A detection algorithm $\calA_n$ for the problem in Definition~\ref{def:PDSquery}, makes up to $\s{Q}$ non-adaptive edge queries, and based on the query responses is tasked with outputting a decision in $\{0,1\}$. We define the \emph{risk} of a detection algorithm $\calA_n$ as the sum of its $\s{Type}$-$\s{I}$ and $\s{Type}$-$\s{II}$ errors probabilities, namely,
\begin{align}
\s{R}(\calA_n) = \pr_{\calH_0}(\calA_n(\s{G}_n)=1)+\pr_{\calH_1}(\calA_n(\s{G}_n)=0),
\end{align}
where $\pr_{\calH_0}$ and $\pr_{\calH_1}$ denote the probability distributions under the null and the alternative hypothesis, respectively. If $\s{R}(\calA_n)\to0$ as $n\to\infty$, then we say that $\calA_n$ solves the detection problem. Our primary goals in this paper are:
\begin{itemize}
\item To characterize the statistical limits of $\s{QPDS}(n,k,p,q,\s{Q})$, namely, to derive necessary and sufficient conditions for when its \emph{statistically impossible} and \emph{statistically possible} to solve the detection problem, ignoring algorithmic computational constraints. 
\item To devise efficient polynomial-time algorithms for $\s{QPDS}(n,k,p,q,\s{Q})$.
\end{itemize}

\subsection{Related Work and Main Contributions}

The problem of finding cliques in an Erd\H{o}s-R\'{e}nyi random graph under the same edge query model was considered in \cite{Feige20}. It was shown that under certain limitations on the adaptivity of the considered class of algorithms, any algorithm that makes $\s{Q} = O(n^\alpha)$ adaptive edge queries, with $\alpha<2$, in $\ell$ rounds finds cliques of size at most $(2-\epsilon)\log_2n$ where $\epsilon= \epsilon(\alpha,\ell)>0$. This lower bound should be contrasted with the fact that current state-of-the-art algorithms that make $\s{Q} = O(n^\alpha)$ queries find a clique of size approximately $(1+\alpha/2)\log_2n$. This result was later improved in \cite{Alweiss2020}, where the dependency of $\epsilon$ on $\ell$ was removed. Closing the gap between those bounds seems to be a challenging open problem. We also mention \cite{Ferber15,Ferber17,Conlon18}, which study the problems of finding a Hamilton cycle, long paths, and  a copy of a fixed target graph, in sparse random graphs under the adaptive edge query model. Another recent line of active research is the analysis of the query complexity in certain clustering tasks, such as, the stochastic block model and community detection \cite{Mazumdar2017ClusteringWN,Mazumdar2017QueryCO,Vinayak2016CrowdsourcedCQ,Hartmann2016ClusteringEN,Anagnostopoulos2016CommunityDO}.

Most closely related to our paper are \cite{racz2020finding,mardia2020finding}, where the special case of the planted clique model, where $p = 1$ and $q=1/2$, under the adaptive edge query model, was considered. For this model, assuming unbounded computational resources, upper and lower bounds on the query complexity for both detecting and recovery were established. Specifically, it was shown in \cite{racz2020finding} that no algorithm that makes at most $\s{Q} = o(n^2/k^2)$ adaptive queries to the adjacency matrix of $\s{G}_n$ is likely to solve the detection problem. On the other hand, when $k\geq (2 + \epsilon) \log_2 n$, for any $\epsilon>0$, it was shown in \cite{racz2020finding} that there exists an algorithm (not
polynomial time) that solves the detection problem by making at least
$\s{Q} = (2+\epsilon)\frac{n^2}{k^2}\log_2^2n$ adaptive queries. For the recovery task, it was shown in \cite{racz2020finding} that no algorithm that makes at most $\s{Q} = o(n^2/k^2+n)$ adaptive queries exists, while recovery is possible using $\s{Q} = o(n^2/k^2\log_2^2n+n\log_2n)$ adaptive queries. Note that when the whole graph is shown to the algorithm, namely, $\s{Q}=\binom{n}{2}$, then the above detection upper bound boils down to $k>(2+\epsilon)\log_2n$, which is folklore and well-known to be tight. On the other hand, the above detection lower bound gives $k=O(1)$, which is loose. Sub-linear time algorithms that find the planted clique in the regime $k=\omega(\sqrt{n\log\log n})$ were proposed in \cite{mardia2020finding}. Specifically, among other things, it was shown that a simple and efficient algorithm can detect the planted clique using $\s{Q} = O(\frac{n^3}{k^3}\log^3n)$ non-adaptive queries; conversely using the planted clique conjecture, it was shown that a certain class of non-adaptive algorithms cannot detect the planted clique if $\s{Q} = o(\frac{n^3}{k^3})$, suggesting that in the regime where $\frac{n^2}{k^2}\ll\s{Q}\ll\frac{n^3}{k^3}$ polynomial-time algorithms do not exist.

In this paper, we generalize and strengthen the results of \cite{racz2020finding}, resolving two open questions raised in the same paper. First, we consider the more general $\s{PDS}$ model which allows for arbitrary edge probabilities. While this might seem as a rather incremental contribution, it turns out that the lower bounding techniques used in \cite{racz2020finding} are quite weak and result in loose bounds on the query complexity for the $\s{PDS}$ model. In a nutshell, the main observation in the proof of the lower bound in \cite{racz2020finding} is that if $\s{Q}\ll n^2/k^2$, then with high probability all edge queries will fall outside the planted clique, no matter how strong/sophisticated the query mechanism is. Therefore, detection would be impossible. While the same bound holds for the $\s{PDS}$ model as well, it does not capture the intrinsic dependency of the query complexity on the edge densities $p$ and $q$. More importantly, as was mentioned above, even for the planted clique problem, the results of \cite{racz2020finding} exhibit a $\s{polylog}$ gap between the upper and lower bounds on the query complexity. We close this gap by providing asymptotically tight bounds. Specifically, we show that an algorithm that must makes $\s{Q} = \Omega(\frac{n^2}{k^2\chi^4(p||q)}\log^2n)$ non-adaptive queries to the adjacency matrix of the graph to be able to detect the planted subgraph, where $\chi^2(p||q)$ is the chi-square distance. On the other hand, we devise a quasi-polynomial-time combinatorial algorithm that detects the planted subgraph with high probability by making $\s{Q} = O(\frac{n^2}{k^2\chi^4(p||q)}\log^2n)$ non-adaptive queries. For the lower bound, we derive first high probability lower and upper bounds on the number of edge queries $\mathscr{C}$ that fall inside the planted subgraph, associated with the optimal query mechanism. Then, we develop a general information-theoretic lower bound on the risk of any algorithm that is given those $\s{Q}$ queries, and essentially observes a subgraph of the $\s{PDS}$ model, with a planted signal that is an arbitrary sub-structure of the original planted subgraph on $\mathscr{C}$ edges. We also propose a polynomial-time algorithm which is able to detect the planted subgraph using $\s{Q} = \Omega(\frac{n^3}{k^3\chi^2(p||q)}\log^3 n)$ queries. In the leftover regime, where $\frac{n^2}{k^2}\ll\s{Q}\ll \frac{n^3}{k^3}$ and $k=\omega(\sqrt{n})$, we conjecture that no polynomial-time algorithms exist for detection. Finally, as we discuss later in the paper, the generality of our techniques allows for arbitrary \emph{planting} and \emph{noise} distribution $\calP$ and $\calQ$, respectively, where the $\s{PDS}$ model boils down to $\calP=\s{Bern}(p)$ and $\calQ=\s{Bern}(q)$.


\subsection{Main Results}

The following theorem determines (up to a constant factor) the number of queries necessary to solve the query-$\s{PDS}$ detection problem. Recall that $\chi^2(p||q)$ and $d_{\s{KL}}(p||q)$ denote the chi-square distance and the Kullback-Leibler (KL) divergence between two $\s{Bern}(p)$ and $\s{Bern}(q)$ random variables, respectively. Note that $\chi^2(p||q) = \frac{(p-q)^2}{q(1-q)}$. 

\begin{theorem}[Detecting a planted densest subgraph]\label{thm:1}  Consider the $\s{QPDS}(n,k,p,q)$ detection problem, and let $\epsilon>0$ be arbitrary. The following statements hold.
\begin{enumerate}
\item (Non-adaptive lower-bound) The risk of any algorithm $\calA_n$ that makes at most $\s{Q}$ non-adaptive edge queries, is $\s{R}(\calA_n)\geq 1-o(1)$, if
\begin{align}
\s{Q} < (2-\epsilon)\cdot\frac{n^2}{k^2\chi^4(p||q)}\log^2 \frac{n}{k}.
\label{eqn:ITlimit}
\end{align}
\item (Statistical sufficiency) Suppose that $k\geq(2+\epsilon_0)\frac{\log n}{d_{\s{KL}}(p||q)}$, for some $\epsilon_0>0$. 
It is possible to detect the presence of a planted densest subgraph, i.e, $\s{R}(\calA_n)\leq o(1)$, by making
\begin{align}
\s{Q} \geq (2+\epsilon)\cdot\frac{n^2}{k^2d_{\s{KL}}^2(p||q)}\log^2 \frac{n}{k}\label{eqn:scanstat}
\end{align}
queries. Moreover, the queries can be non-adaptive.
\item (Computational sufficiency) Suppose that $k=\Omega(\sqrt{n\log n/\chi^2(p||q)})$. There exists a polynomial-time algorithm $\calA_n$ that can detect the presence of a planted densest subgraph, i.e, $\s{R}(\calA_n)\leq o(1)$, by making
\begin{align}
\s{Q} = O\p{\frac{n^3}{k^3\chi^2(p||q)}\log^3n}\label{eqn:countstat}
\end{align}
queries. Moreover, the queries can be non-adaptive.
\end{enumerate}
\end{theorem}
The proof of Theorem~\ref{thm:1} is given in Sections~\ref{sec:upper_bound} and \ref{sec:lowerBound}. Let us describe briefly the algorithms achieving the query complexities in the second and third items of Theorem~\ref{thm:1}. The test achieving the query complexity in the second item of Theorem~\ref{thm:1} is the \emph{scan test}. In the first step of this algorithm we subsample a set $\calS\subset[n]$ of $\s{M}\in\mathbb{N}$ elements, drawn uniformly at random, and then take all pairwise queries among those elements, resulting in $\s{Q}=\binom{\s{M}}{2}$ non-adaptive queries. Observe that at the end of this first step we, in fact, observe the subgraph $\s{G}_{\calS}$ of $\s{G}_n$ induced by $\calS$. Now, given the responses to those queries, in order to distinguish between hypotheses $\calH_0$ and $\calH_1$, we search for the densest subgraph (in the sense of the number of active edges) of a certain size in $\s{G}_{\calS}$, and then compare the result to a carefully chosen threshold. The test achieving the query complexity in the third item of Theorem~\ref{thm:1}, on the other hand, is the \emph{degree test} which, roughly speaking, counts the number of ``high degree" vertices (i.e., vertices with degrees exceeding a certain threshold) in a randomly chosen induced subgraph of $\s{G}_n$, and then compare the result to a certain threshold. A very similar variant of this degree test was proposed and analyzed in \cite{mardia2020finding}, for the planted clique problem. It is clear that the combinatorial scan test is computationally expensive (super-polynomial), while the degree test has a polynomial-time complexity. Interestingly, adaptivity is not needed in order to achieve the statistical barrier. It should be emphasized that the condition $k\geq(2+\epsilon_0)\frac{\log n}{d_{\s{KL}}(p||q)}$ in the second part of Theorem~\ref{thm:1} is essential because, otherwise, if $k<(2-\epsilon_0)\frac{\log n}{d_{\s{KL}}(p||q)}$ detection is known to be statistically impossible even if we observe/query the whole graph.

Note that the lower and upper bound for the non-adaptive case in \eqref{eqn:ITlimit} and \eqref{eqn:scanstat} are tight up to a constant factor; the former depend on $p$ and $q$ through the chi-square distance, while the latter through the KL divergence. As we show in the proof, an alternative condition for the scan test to succeed in detection is
\begin{align}
\s{Q} \geq (2+\epsilon)\cdot\frac{Cn^2}{k^2\chi^4(p||q)}\log^2 \frac{n}{k},
\end{align}
for some constant $C\geq 8$; the above condition meets the lower bound in \eqref{eqn:ITlimit} up to the constant factor $C$. We strongly believe that the source of this negligible gap is due to the lower bound, namely, the $\chi^4(p||q)$ factor in \eqref{eqn:ITlimit} can be in fact replaced with $d_{\s{KL}}^2(p||q)$. It should be emphasized, however, that in the special case of planted clique where $p=1$ and $q=1/2$, we have $\chi^4(p||q) = d_{\s{KL}}^2(p||q)=1$, and thus our bounds in \eqref{eqn:ITlimit} and \eqref{eqn:scanstat} are tight and fully characterize the statistical limit of detection. This closes the gap in \cite{racz2020finding}. 

As can be noticed from the second and third items of Theorem~\ref{thm:1}, there is a significant gap between the query complexity of the optimal algorithm and that of the computationally efficient one. This observation raises the following intriguing question: \emph{what is the sharp condition on $(n,k,p,q,\s{Q})$ under which the problem admits a computationally efficient test with vanishing risk, and conversely, without which no algorithm can detect the planted dense subgraph reliably in polynomial-time?} The gap observed in our problem is common to many contemporary problems in high-dimensional statistics studied over the last few years. Indeed, recently, there has been a success in developing a rigorous notion of what can and cannot be achieved by efficient algorithms. Specifically, a line of work initiated in \cite{berthet2013complexity} has aimed to explain these statistical-computational gaps by reducing from conjecturally hard average-case problems in computer science, most notably, the planted clique problem, conjectured to be computationally hard in the regime $k=o(\sqrt{n})$. Accordingly, such reductions from planted clique were established to prove tight statistical-computational gaps for a wide verity of detection and recovery problems, e.g., \cite{berthet2013complexity,ma2015computational,cai2015computational,hajek2015computational,wang2016average,wang2016statistical,gao2017sparse,brennan2018reducibility,Brennan2019UniversalityOC,wu2018statistical,brennan20a}.

As mentioned above, it is widely believed that the planted clique detection problem cannot be solved in randomized polynomial time when $k=o(\sqrt{n})$, which we shall refer to as the \emph{planted clique conjecture}, stated as follows. Below, we let $(\calH_0^{\s{PC}},\calH_1^{\s{PC}})$ denote the planted clique detection problem, and we recall that it is a special case of the $\s{PDS}$ detection problem with edge probabilities $p=1$ and $q=1/2$; for simplicity of notation we designate $\s{PC}(n,k)=\s{PDS}(n,k,1,1/2)$.
\begin{conjecture}[Planted clique conjecture]\label{conj:PC}
Suppose that $\{\calA_n\}$ is a sequence of randomized polynomial-time algorithms $\calA_n$ and $\{k_n\}$ is a sequence of positive integers
satisfying that $\limsup_{n\to\infty}\frac{\log k_n}{\log n}<1/2$. Then, if $\s{G}_n$ is an instance of $\s{PC}(n,k)$, it holds that
\begin{align}
\liminf_{n\to\infty}\pp{\pr_{\calH_0^{\s{PC}}}\p{\calA_n(\s{G}_n)=1}+\pr_{\calH_1^{\s{PC}}}\p{\calA_n(\s{G}_n)=0}}\geq 1.
\end{align}
\end{conjecture}
Going back to our problem, it is clear that for $k = o(\sqrt{n})$, and any $1\leq\s{Q}\leq\binom{n}{2}$ solving $\s{QPDS}(n,k,\s{Q},p,q)$ is computationally hard, namely, there exists no randomized polynomial-time \emph{adaptive} algorithm that makes up to $\s{Q}$ edge queries and is able to solve the detection problem. Accordingly, Theorem~\ref{thm:1} and the planted clique conjecture give a partial phase diagram for when detection is statistically impossible, computational hard, and computational easy. The union of the last two regime is the statistically possible regime. Treating $k$ and $\s{Q}$ as polynomials in $n$, i.e., $\s{Q}=\Theta(n^\alpha)$ and $k = \Theta(n^\beta)$, for some $\alpha\in(0,2)$ and $\beta\in(0,1)$, we obtain the phase diagram in Fig.~\ref{fig:spcaphasediagram2}. Specifically,
\begin{enumerate}
    \item \emph{Computationally easy regime (blue region):} there is a polynomial-time algorithm for the detection task when $\alpha>3-3\beta$ and $\beta>1/2$.
    \item \emph{Computationally hard regime (red region):} there is an inefficient algorithm for detection when $\beta<1/2$ and $2-2\beta<\alpha$, but the problem is computationally hard (no polynomial-time algorithm exists) in the sense that it is at least as hard as solving the planted clique problem.
    \item \emph{Conjecturally hard regime (black region):} there is an inefficient algorithm for detection when $2-2\beta<\alpha<3-3\beta$, but we conjecture that there is no polynomial-time algorithm. This was also conjectured in \cite{mardia2020finding}.
    \item \emph{Statistically impossible regime:} the task is statistically/information-theoretically impossible when $\alpha<2-2\beta$.
\end{enumerate}

\begin{figure}[t!]
\centering

\begin{tikzpicture}[scale=1.5]
\tikzstyle{every node}=[font=\footnotesize]
\def\xmin{0}
\def\xmax{3.5}
\def\ymin{0}
\def\ymax{4.5}

\draw[->] (\xmin,\ymin) -- (\xmax,\ymin) node[right] {$\beta$};
\draw[->] (\xmin,\ymin) -- (\xmin,\ymax) node[above] {$\alpha$};

\node at (3, 0) [below] {$1$};
\node at (1.5, 0) [below] {$\frac{1}{2}$};
\node at (0, 2) [left] {$1$};
\node at (0, 4) [left] {$2$};
\node at (0, 0) [left] {$0$};
\node at (0, 3) [left] {$\frac{3}{2}$};

\filldraw[fill=gray!25, draw=black] (0, 0) -- (3, 0) -- (0, 4) -- (0, 0);
\filldraw[fill=blue!25, draw=black] (3, 0) -- (1.5, 3) -- (1.5, 4) -- (3, 4) -- (3, 0);
\filldraw[fill=red!25, draw=black] (1.5, 2) -- (0, 4) -- (1.5, 4) -- (1.5, 2);
\filldraw[fill=black!65, draw=black] (3, 0) -- (1.5, 2) -- (1.5, 3) -- (3, 0);
\node at (0.8, 1.6) {``$\s{\mathbf{Statistically}}$};
\node at (1.1, 1.3) {$\s{\mathbf{Impossible}}$"};
\node[rotate=-80] at (2.55, 2.5) {``$\mathbf{Computationally}\;\s{\mathbf{Easy}}$"};
\node[rotate=-60] at (0.95, 3.4) {``$\s{\mathbf{Hard}}$"};
\node[rotate=0] at (1.85, 1.9) {\textbf{?}};
\draw [dashed] (1.5,0) -- (1.5,2);
\draw [dashed] (0,3) -- (1.5,3);
\draw [dashed] (0,2) -- (1.5, 2);
\end{tikzpicture}

\caption{Phase diagram for detecting the presence of a planted dense subgraph, as a function of the dense subgraph size $k = \Theta(n^{\beta})$ and the number of non-adaptive edge queries $\s{Q}=\Theta(n^{\alpha})$.}
\label{fig:spcaphasediagram2}
\end{figure}
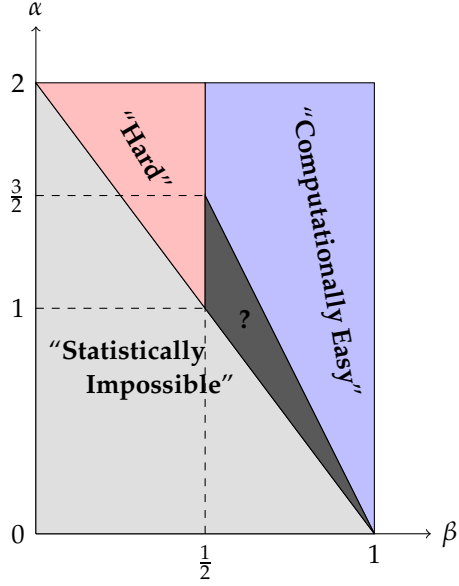

It turns out that our techniques allows for a more general submatrix detection problem with arbitrary planting and noise distribution $\calP$ and $\calQ$, respectively, defined as follows. 
\begin{definition}[General submatrix detection]\label{eqn:SD}
Given a pair of distributions $(\calP,\calQ)$ over a measurable space $(\calX,\mathcal{B})$, let $\s{SD}(n,k,\calP,\calQ)$ denote the hypothesis testing problem with observation $\s{X} \in \calX^{n \times n}$ and hypotheses
\begin{align}
\calH_0: \s{X} \sim \calQ^{\otimes n \times n}\quad\quad\s{vs.}\quad\quad\calH_1: \s{X}\sim\mathcal{D}(n,k,\calP, \calQ),\label{eqn:modelSD}
\end{align}
where $\mathcal{D}(n, k, \calP, \calQ)$ is the distribution of symmetric matrices $\s{X}$ with entries $\s{X}_{ij} \sim \calP$ if $i, j \in \calK$ and $\s{X}_{ij} \sim \calQ$ otherwise that are conditionally independent given $\calK$, which is chosen uniformly at random over all $k$-subsets of $[n]$.
\end{definition}
Now, consider the following natural generalization of the edge query oracle in Definition~\ref{def:oracleSD} to the above submatrix problem.
\begin{definition}[Oracle/Entries queries]\label{def:oracleSD}
Consider an $n\times n$ symmetric matrix $\s{X}\in\mathbb{R}^{n\times n}$. An oracle $\calO_{\s{sub}}:[n]\times[n]\to\mathbb{R}$, takes as input a pair of vertices $i,j\in[n]\times[n]$, and outputs $\calO_{\s{sub}}(i,j)=\s{X}_{ij}$ in response.
\end{definition}
Then, we define the query-submatrix detection problem as follows.
\begin{problem}[Query-submatrix detection problem]\label{def:SDSquery}
Consider the submatrix detection problem in Definition~\ref{eqn:SD}.  There is an oracle $\calO_{\s{sub}}$ as defined in Definition~\ref{def:oracleSD}. Find a set of queries $\mathbb{Q}\subseteq[n]\times[n]$ such that
$\s{Q}=|\mathbb{Q}|$, and from the oracle answers it is possible to solve the detection problem in \eqref{eqn:modelSD}. Henceforth, we denote this detection problem by $\s{QSD}(n,k,\calP,\calQ,\s{Q})$.
\end{problem}
It is clear that the $\s{PDS}$ model correspond to $\calP=\s{Bern}(p)$ and $\calQ = \s{Bern}(q)$, while $\s{X}$ is the graph adjacency matrix. Then, it can be shown that under mild conditions on the tails of the distribution $\calP$ and $\calQ$, Theorem~\ref{thm:1} holds for the setting in Definition~\ref{eqn:SD} with $\chi^2(p||q)$ and $d_{\s{KL}}(p||q)$ replaced by $\chi^2(\calP||\calQ)$ and $d_{\s{KL}}(\calP||\calQ)$, respectively. Specifically, the lower bound in Theorem~\ref{thm:1} holds whenever $0<\chi^2(\calP||\calQ)<\infty$, and the upper bounds hold if, for example, the log-likelihood ratio $\calL\triangleq\log\frac{\mathrm{d}\calP}{\mathrm{d}\calQ}$ is sub-Gaussian, or, even slightly weaker; if there is a constant $C \ge 1$ such that
\begin{subequations}
\begin{align}
\psi_{\calP}(\lambda) - d_{\s{KL}}( \calP \| \calQ) \cdot \lambda &\le C \cdot d_{\s{KL}}( \calP \| \calQ) \cdot \lambda^2 \quad \forall \lambda \in [-1, 0], \\
\psi_{\calQ}(\lambda) + d_{\s{KL}}( \calQ \| \calP) \cdot \lambda &\le C \cdot d_{\s{KL}}( \calQ \| \calP) \cdot \lambda^2 \quad 
\forall \lambda \in [-1, 1],
\end{align}\label{eqn:PDSGenassumption}
\end{subequations}
where $\psi_{\calQ}\triangleq\log\bE_{\calQ}[\exp(\lambda\calL)]$ and $\psi_{\calP}\triangleq\log\bE_{\calP}[\exp(\lambda\calL)]$. Note that the above assumption was considered also in \cite{Hajek17,Brennan2019UniversalityOC}, and arise naturally from classical binary hypothesis testing, in order to control the tails of the error probabilities associated with the simple tests we propose. 

\section{Algorithms and Upper Bounds}\label{sec:upper_bound}

In this section we prove items 2 and 3 in Theorem~\ref{thm:1}. To that end, we propose two algorithms whose performance match \eqref{eqn:scanstat}--\eqref{eqn:countstat}. Below, we denote the adjacency matrix of the underlying graph $\s{G}_n$ by $\s{A}\in\ppp{0,1}^{n\times n}$, with its $(i,j)$ entry denoted by $\s{A}_{ij}$, for any $1\leq i,j\leq n$.

\subsection{Scan Test}\label{subsec:scan}

In this subsection we analyzed the scan test in Algorithm~\ref{algo:0}. The parameters $\s{M}$, $\s{N}_0$, and $\tau_{\s{scan}}$ in Algorithm~\ref{algo:0} will be specified below. In the first step of the scan test we subsample a set $\calS\subseteq[n]$ of $\s{M}\in\mathbb{N}$ elements, drawn uniformly at random, and take all pairwise queries among these elements. Therefore, the number of queries is $\s{Q}=\binom{\s{M}}{2}$. Given these queries we, in fact, learn the induced subgraph, denoted by $\s{G}_{\calS}$, on the set of vertices $\calS$. Then, using this subgraph, we wish to distinguish between $\calH_0$ and $\calH_1$. Recall that $\calK$ denotes the set of vertices over which the densest subgraph was planted under $\calH_1$, and let $\s{N}$ denote the number of planted dense subgraph vertices in $\calS$, i.e., $\s{N} \triangleq |\calK\cap\calS|$. Since $\calK$ and $\calS$ are drawn uniformly at random from all sets of size $k$ and $\s{M}$ from $[n]$, respectively, we observe that $\s{N}\sim\s{Hypergeometric}(n,k,\s{M})$, namely, $\s{N}$ has a Hypergeometric distribution with parameters $n$, $k$, and $\s{M}$. Accordingly, we have that $\bE(\s{N}) = \frac{k\s{M}}{n}$, and $\s{var}(\s{N})\leq \frac{k\s{M}}{n}\cdot(1-k/n)$. Therefore, Chebyshev's inequality implies that
\begin{align}
\pr\pp{\s{N}\leq(1-\epsilon)\bE(\s{N})}\leq\frac{1}{\epsilon^2\bE(\s{N})}.\label{eqn:sNconcen}
\end{align}
\begin{algorithm}[t]
\caption{\texttt{Scan Test}\label{algo:0}}
\footnotesize
\begin{algorithmic}[1]
\REQUIRE $\s{G}_n$, $\s{Q}=\binom{\s{M}}{2}$, $\s{N}_0 = (1-\epsilon)\frac{k\s{M}}{n}$, and $\tau_{\s{scan}}=\binom{\s{N}_0}{2}\cdot\gamma$, for $\epsilon\in(0,1)$ and $\gamma\in[q,p]$.
\STATE Subsample a set $\calS$ of $\s{M}$ elements drawn uniformly at random from $[n]$.
\STATE Take all pairwise queries among the elements in $\calS$, and obtain $\s{A}_{ij}$, for all $i,j\in\calS$.
\STATE Compute 
$$
\s{S_{scan}}\triangleq\max_{\calL\subset\calS:|\calL|=\s{N}_0}\sum_{i<j\in\calL}\s{A}_{ij}.
$$
\STATE If $\s{S_{scan}}>\tau_{\s{scan}}$ decide $\calH_1$; otherwise, decide $\calH_0$. 
\end{algorithmic}
\end{algorithm}
Thus, provided that $\epsilon^2\bE(\s{N})\to\infty$, we see that with probability tending to unity $\s{N}\geq(1-\epsilon)\frac{k\s{M}}{n}\triangleq\s{N}_0$. The implication of this is the following: as mentioned above we are tasked with the following detection problem:
\begin{align}
    \calH_0': \s{G}_{\calS}\sim\calG(\s{M},q)\quad\quad\s{vs.}\quad\quad\calH_1': \s{G}_{\calS}\sim\calG(\s{M},\s{N},p,q).\label{eqn:model_reduced}
\end{align}
Therefore, \eqref{eqn:model_reduced} represents a $\s{PDS}(\s{M},\s{N},p,q)$ detection problem, but the size of the planted densest subgraph is random. Nonetheless, due to \eqref{eqn:sNconcen}, it should be clear that by replacing $\s{N}$ with $\s{N}_0$ in \eqref{eqn:model_reduced}, the detection problem becomes algorithmically harder; thus upper bounds on $\s{PDS}(\s{M},\s{N}_0,p,q)$ imply corresponding upper bounds on $\s{PDS}(\s{M},\s{N},p,q)$. Below, we prove this rigorously.

Recall that $\s{A}$ is the adjacency matrix of $\s{G}_n$. The scan test is defined as follows
\begin{align}
\calA_{\s{scan}}(\s{A}_{\calS})\triangleq\Ind\ppp{\max_{\calL\subset\calS:|\calL|=\s{N}_0}\sum_{i<j\in\calL}\s{A}_{ij}\geq\tau_{\s{scan}}},
\end{align}
where $\tau_{\s{scan}}\triangleq\binom{\s{N}_0}{2}\cdot\gamma$, for some $\gamma\in[q,p]$. Under the null hypothesis, for any fixed subset $\calL$ of size $\s{N}_0$, it is clear that $\sum_{i<j\in\calL}\s{A}_{ij}\sim\s{Binomial}\p{\binom{\s{N}_0}{2},q}$.
By the union bound and Chernoff's inequality,
\begin{align}
\pr_{\calH_0'}\p{\calA_{\s{scan}}(\s{A}_{\calS})=1} &= \pr_{\calH_0'}\pp{\max_{\calL\subset\calS:|\calL|=\s{N}_0}\sum_{i<j\in\calL}\s{A}_{ij}\geq\tau_{\s{scan}}}\\
&\leq \sum_{\calL\subset\calS:|\calL|=\s{N}_0}\pr_{\calH_0'}\pp{\sum_{i<j\in\calL}\s{A}_{ij}\geq\tau_{\s{scan}}}\\
&\leq \binom{\s{M}}{\s{N}_0}\cdot\pr\pp{\s{Binomial}\p{\binom{\s{N}_0}{2},q}\geq\tau_{\s{scan}}}\\
&\leq \p{\frac{e\s{M}}{\s{N}_0}}^{\s{N}_0}\exp\p{-\binom{\s{N}_0}{2}d_{\s{KL}}(\gamma||q)}\\
& =\exp\p{\s{N}_0\log\frac{e\s{M}}{\s{N}_0}-\binom{\s{N}_0}{2}d_{\s{KL}}(\gamma||q)}.\label{eqn:Type1Scan_Chernoff}
\end{align}
Under the alternative hypothesis, conditioned on $\s{N}=\s{N}'$, for some $\s{N}'\geq \s{N}_0$, the scan test statistics $\max_{\calL\subset\calS:|\calL|=\s{N}_0}\sum_{i<j\in\calL}\s{A}_{ij}$ stochastically dominates $\s{Binomial}\p{\binom{\s{N}_0}{2},p}$. By the multiplicative Chernoff's bound,
\begin{align}
\pr_{\calH_1'}\p{\calA_{\s{scan}}(\s{A}_{\calS})=0} &= \pr_{\calH_1'}\pp{\max_{\calL\subset\calS:|\calL|=\s{N}_0}\sum_{i<j\in\calL}\s{A}_{ij}\leq\tau_{\s{scan}}}\\
&\leq \pr\p{\s{N}\leq \s{N}_0}+\pr\pp{\s{Binomial}\p{\binom{ \s{N}_0}{2},p}\leq\tau_{\s{scan}}}\\
&\leq \frac{1}{\epsilon^2\bE\s{N}_0}+\exp\p{-\frac{\binom{\s{N}_0}{2}p}{2}\p{1-\frac{\gamma}{p}}^2},\label{eqn:Type2Scan_Chernoff}
\end{align}
which clearly goes to $0$ since $\s{N}_0\to\infty$. Therefore, combining \eqref{eqn:Type1Scan_Chernoff} and \eqref{eqn:Type2Scan_Chernoff} we see that $\s{R}(\calA_{\s{scan}})\leq\delta$, if
\begin{align}
    d_{\s{KL}}(\gamma||q)>\frac{2}{\s{N}_0}\log\frac{\s{M}}{\s{N}_0}+\frac{2}{\s{N}_0^2}\log\frac{1}{\delta} \geq \frac{\sqrt{2}n}{\sqrt{\s{Q}}k}\log\frac{n}{k}+\frac{n^2}{\s{Q}k^2}\log\frac{1}{\delta}.
\end{align}
By taking $\gamma$ arbitrary close to $p$, we get the query complexity bound in \eqref{eqn:scanstat}, for arbitrary $\epsilon>0$. Also, note that the condition $k\geq(2+\epsilon_0)\frac{\log n}{d_{\s{KL}}(p||q)}$, for some $\epsilon_0>0$, in the second part of Theorem~\ref{thm:1} follows from the constraint that $\s{M}\leq n$. As we mentioned right after Theorem~\ref{thm:1}, a weaker bound exhibiting a dependency on the chi-square distance can be derived. Specifically, let $\tau_{\s{scan}} = \binom{\s{N}_0}{2}\frac{p+q}{2}$. By the union bound and Bernstein's inequality,
\begin{align}
\pr_{\calH_0'}\p{\calA_{\s{scan}}(\s{A}_{\calS})=1} &= \pr_{\calH_0'}\pp{\max_{\calL\subset\calS:|\calL|=\s{N}_0}\sum_{i<j\in\calL}\s{A}_{ij}\geq\tau_{\s{scan}}}\\
&\leq \sum_{\calL\subset\calS:|\calL|=\s{N}_0}\pr_{\calH_0'}\pp{\sum_{i<j\in\calL}\s{A}_{ij}\geq\tau_{\s{scan}}}\\
&\leq \binom{\s{M}}{\s{N}_0}\cdot\pr\pp{\s{Binomial}\p{\binom{\s{N}_0}{2},q}\geq\tau_{\s{scan}}}\\
&\leq \p{\frac{e\s{M}}{\s{N}_0}}^{\s{N}_0}\exp\p{-\frac{\binom{\s{N}_0}{2}^2(p-q)^2/4}{2\binom{\s{N}_0}{2}q+\binom{\s{N}_0}{2}(p-q)/3}}\\
& =\exp\p{\s{N}_0\log\frac{e\s{M}}{\s{N}_0}-C\s{N}_0^2\frac{(p-q)^2}{q(1-q)}},\label{eqn:Type1Scan}
\end{align}
for some constant $C>8$, and note that $\frac{(p-q)^2}{q(1-q)}=\chi^2(p||q)$. Similarly to \eqref{eqn:Type2Scan_Chernoff}, under the alternative hypothesis, conditioned on $\s{N}=\s{N}'$, for some $\s{N}'\geq \s{N}_0$, the scan test statistics $\max_{\calL\subset\calS:|\calL|=\s{N}_0}\sum_{i<j\in\calL}\s{A}_{ij}$ stochastically dominates $\s{Binomial}\p{\binom{\s{N}_0}{2},p}$. By the multiplicative Chernoff's bound,
\begin{align}
\pr_{\calH_1'}\p{\calA_{\s{scan}}(\s{A}_{\calS})=0} &= \pr_{\calH_1'}\pp{\max_{\calL\subset\calS:|\calL|=\s{N}_0}\sum_{i<j\in\calL}\s{A}_{ij}\leq\tau_{\s{scan}}}\\
&\leq \pr\p{\s{N}\leq \s{N}_0}+\pr\pp{\s{Binomial}\p{\binom{ \s{N}_0}{2},p}\leq\tau_{\s{scan}}}\\
&\leq \frac{1}{\epsilon^2\bE\s{N}_0}+\exp\p{-\frac{\p{2\binom{\s{N}_0}{2}-\binom{\s{N}_0}{2}}^2(p-q)^2}{8\binom{\s{N}_0}{2}p}}\\
& =\frac{1}{\epsilon^2\bE\s{N}_0}+\exp\p{-C\s{N}_0^2q},\label{eqn:Type2Scan}
\end{align}
which converges to $0$. Therefore, combining \eqref{eqn:Type1Scan} and \eqref{eqn:Type2Scan} we see that $\s{R}(\calA_{\s{scan}})\leq\delta$, if
\begin{align}
    \chi^2(p||q)\geq\Omega\p{\frac{1}{\s{N}_0}\log\frac{\s{M}}{\s{N}_0}+\frac{1}{\s{N}_0^2}\log\frac{2}{\delta}} = \Omega\p{\frac{n}{\sqrt{\s{Q}}k}\log\frac{n}{k}+\frac{n^2}{\s{Q}k^2}\log\frac{2}{\delta}}.
\end{align}


\subsection{Degree Test}

In this subsection we analyze the degree test in Algorithm~\ref{algo:1}. The parameters $n'$, $\s{M}$, and $\s{N}_0$ in Algorithm~\ref{algo:1} will be specified below. Specifically, let $\calS$ denote a set of $n'\in\mathbb{N}$ vertices drawn uniformly at random from $[n]$, and $\s{G}_{\calS}$ be the subgraph of $\calG$ induced by this set. Then, we subsample $\s{M}\in\mathbb{N}$ elements from $\calS$, and denote those set of elements by $\calM$.
\begin{algorithm}[t]
\caption{\texttt{Degree Test}\label{algo:1}}
\footnotesize
\begin{algorithmic}[1]
\REQUIRE $\s{G}_n$, $n'=\Omega\p{\frac{n^2\log n}{k^2\chi^2(p||q)}}$, $\s{M}=\Omega\p{\frac{n}{k}\log^2 n}$, and $\s{N}_0=(1-\epsilon)\frac{kn'}{n}$, for $\epsilon\in(0,1)$.
\STATE Let $\calS$ denote a set of $n'$ vertices drawn uniformly at random from $[n]$, and $\s{G}_{\calS}$ be the subgraph of $\calG$ induced by $\calS$. 
\STATE Subsample $\s{M}$ elements uniformly at random from $\calS$, and denote those set of elements by $\calM$.
\STATE Compute 
$$
\s{C_{deg}}\triangleq\sum_{i\in\calM}\Ind\pp{\sum_{j\in\calS}\s{A}_{ij}>n'q+\frac{\s{N}_0(p-q)}{2}}.
$$
\STATE If $\s{C_{deg}}>2\log n'$ decide $\calH_1$; otherwise, decide $\calH_0$. 
\end{algorithmic}
\end{algorithm}
Our test is defined as follows, 
\begin{align}
\calA_{\s{deg}}(\s{A}_{\calS})\triangleq\Ind\ppp{\sum_{i\in\calM}\Ind\pp{\sum_{j\in\calS}\s{A}_{ij}>\tau_{\s{deg}}}\geq 2\log n'},
\end{align}
where $\tau_{\s{deg}}\triangleq n'q+\frac{\s{N}_0(p-q)}{2}$, and $\s{N}_0\in\mathbb{N}$. As in the previous subsection, under $\calH_1$, let $\s{N}$ denote the number of planted dense subgraph vertices in $\calS$, i.e., $\s{N} \triangleq |\calK\cap\calS|$, and note that $\s{N}\sim\s{Hypergeometric}(n,k,n')$. We have shown that with probability tending to unity $\s{N}\geq(1-\epsilon)\frac{kn'}{n}\triangleq\s{N}_0$. Let this event be denoted by $\calC_1$, and so $\pr(\calC_1)\to1$. Next, we analyze the $\s{Type}$-$\s{I}$+$\s{II}$ error probability associated with the above test. Let us start with the null hypothesis. For any $i\in\calM$, we let $\s{X}_i\triangleq\Ind\pp{\sum_{j\in\calS}\s{A}_{ij}>n'q+\frac{\s{N}_0(p-q)}{2}}\sim\s{Bern}\p{\pr\pp{\sum_{j\in\calS}\s{A}_{ij}>n'q+\frac{\s{N}_0(p-q)}{2}}}$. Under $\calH_0$, it is clear that $\sum_{j\in\calS}\s{A}_{ij}\sim\s{Binom}\p{n',q}$, and thus by Bernstein’s inequality
\begin{align}
   \pr_{\calH_0}\pp{\sum_{j\in\calS}\s{A}_{ij}>n'q+\frac{\s{N}_0(p-q)}{2}} \leq \exp\pp{-\Omega\p{\frac{\s{N}_0^2}{n'}\chi^2(p||q)}}.
\end{align}
Accordingly, if $n'< \frac{\s{N}_0^2\chi^2(p||q)}{C\log n}$, which implies that $n'>\Omega\p{\frac{n^2\log n}{k^2\chi^2(p||q)}}$ we will have $\pr_{\calH_0}\pp{\sum_{j\in\calS}\s{A}_{ij}>n'q+\frac{\s{N}_0(p-q)}{2}} \leq n^{-2}$, for some $C>0$. This implies that under $\calH_0$, each $\s{X}_i$ is stochastically dominated by $\s{Bern}(n^{-2})$. Consider now the alternative hypothesis. Recall that conditioned on $\calC_1$, over the induced subgraph $\s{G}_{\calS}$ with $n'$ vertices there are at least $\s{N}_0$ vertices from the planted set. Next, we show that by subsampling $\s{M}$ vertices $\calM$ from $\s{G}_{\calS}$ (as in the second step of Algorithm~\ref{algo:1}), with high probability among those $\s{M}$ samples at least $3\log n'$ vertices fall inside the planted set. To that end, let us divide the subsampled planted vertices $\calK\cap\calS$ into $3\log n'$ disjoint sets $\{\calS_i\}_{i=1}^{3\log n'}$ of equal size $\frac{\s{N}_0}{3\log n'}$. Let $\calE$ denote the event that there exist a set $\calS_i$ which do not intersect $\calM$, namely, $\calE\triangleq\bigcup_{i=1}^{3\log n'}\ppp{\calS_i\cap\calM=\emptyset}$. Then, note that
\begin{align}
    \pr\pp{\calE}&\leq (3\log n')\cdot\pr\pp{\calS_1\cap\calM=\emptyset}\\
    &= (3\log n')\cdot\pp{1-\frac{\s{N}_0}{3n'\log n'}}^{\s{M}}\\
    &\leq (3\log n')\cdot \exp\p{-\frac{\s{N}_0\s{M}}{3n'\log n'}},\label{eqn:donotinter}
\end{align}
namely, the probability that there exists a set $\calS_i$ which do not intersect $\calM$ is at most \eqref{eqn:donotinter}. Thus, taking $\s{M}\geq \frac{6n'\log^2 n'}{\s{N}_0}$, i.e., $\s{M} = \Omega\p{\frac{n}{k}\log^2 n}$ we obtain that $ \pr\pp{\calE}\leq 3\log n'/n'^2$, and accordingly, with probability tending to 1, we sample at least one elements from each $\calS_i$. Therefore, among the $\s{M}$ samples in $\calM$ at least $3\log n'$ vertices fall inside the planted set. Now, for any $i\in\calM\cap(\calK\cap\calS)$, we note that $\sum_{j\in\calS}\s{A}_{ij}$ stochastically dominates $\s{Binom}\p{\s{N}_0,p}+\s{Binom}\p{n-\s{N}_0,q}$, and thus conditioned on $\calE^c$ and $\calC_1$, by the multiplicative Chernoff's bound,
\begin{align}
   \pr_{\calH_1}\pp{\sum_{j\in\calS}\s{A}_{ij}\leq n'q+\frac{\s{N}_0(p-q)}{2}} \leq \exp\pp{-\Omega\p{\frac{\s{N}_0^2}{n'}\chi^2(p||q)}}.
\end{align}
Accordingly, if $n'< \frac{\s{N}_0^2\chi^2(p||q)}{C\log n}$, we will have $\pr_{\calH_1}\pp{\sum_{j\in\calS}\s{A}_{ij}\leq n'q+\frac{\s{N}_0(p-q)}{2}} \leq n^{-2}$, for large enough $C>0$. This implies that under $\calH_1$, conditioned on $\calE^c$ and $\calC$ each $\s{X}_i$ stochastically dominates $\s{Bern}(1-n^{-2})$, for $i\in\calM\cap(\calK\cap\calS)$. Establishing the above results, we are in a position to upper bound the $\s{Type}$-$\s{I}$ and $\s{Type}$-$\s{II}$ error probabilities. Using Markov's inequality, we have
\begin{align}
    \pr_{\calH_0}\p{\calA_{\s{deg}}(\s{A}_{\calS})=1} &= \pr_{\calH_0}\pp{\sum_{i\in\calM}\s{X}_i>2\log n'} \\
    &\leq \pr\pp{\s{Binom}\p{\s{M},n^{-2}}>2\log n'}\\
    &\leq \frac{\s{M}}{2n^2\log n'},
\end{align}
which clearly goes to 0 as $n\to\infty$. On the other hand,
\begin{align}
    \pr_{\calH_1}\p{\calA_{\s{deg}}(\s{A}_{\calS})=0} &= \pr_{\calH_1}\pp{\sum_{i\in\calM}\s{X}_i<2\log n'} \\
    &\leq \pr\pp{\s{Binom}\p{3\log n',1-n^{-2}}<2\log n'}+\pr\p{\calE}+\pr\p{\calC^c}\\
    &\leq e^{-C\log n'}+o(1)\to0,
\end{align}
as $n\to\infty$. Finally, note that the number of queries made by Algorithm~\ref{algo:1} is $\s{Q}=\s{M}\cdot n'>\Omega\p{\frac{n^3}{k^3\chi^2(p||q)}\log^3n}$ and the constraint that $n'<n$ implies that $k\geq \Omega(\sqrt{n\log n/\chi^2(p||q)})$. This concludes the proof.

\section{Statistical Lower Bound}\label{sec:lowerBound}

In this section, we prove the lower bounds in Theorem~\ref{thm:1}, for non-adaptive mechanisms. Our proof consists of two main steps. In the first step, we bound the total number of planted edges any non-adaptive mechanism $\mathscr{Q}_n$ can query. We denote this number of planted queries by $\mathscr{C}(\mathscr{Q}_n)$. Given $\mathscr{Q}_n$, in the second step of the proof, we analyze the resulting detection problem a test $\calA_n$ is faced with, and derive a lower bound on its risk. Note that in order to prove that detection is statistically impossible whenever \eqref{eqn:ITlimit} holds, it is suffice to prove impossibility on the boundary, i.e., when $\s{Q} = \s{Q}^\star\triangleq(2-\epsilon)\cdot\frac{n^2}{k^2\chi^4(p||q)}\log^2 n$, for non-adaptive mechanisms. Indeed, if detection is impossible when $\s{Q} = \s{Q}^\star$, then with less queries $\s{Q}<\s{Q}^\star$ detection remains impossible.

\subsection{Upper Bound on the Query-Number of Planted Edges}

In this subsection, we bound the total number of planted edges any mechanism can query. As mentioned above, this number is denoted by $\mathscr{C}(\mathscr{Q}_n)$ for a given query mechanism $\mathscr{Q}_n$. Letting $\mathbb{Q}$ denote the query trajectory (i.e., the set of $\s{Q}$ queried edges), we note that $\mathscr{C}(\mathscr{Q}_n) = \sum_{\ell=1}^{\s{Q}}\Ind\ppp{\mathbb{Q}_\ell\in\calK\times\calK}$. After making $\s{Q}$ queries, a testing algorithm produces an decision. Specifically, given a trajectory $\mathbb{Q}$ (or, a mechanism $\mathscr{Q}_n$), we denote by $(\calH_0',\calH_1')$ the new hypothesis testing problem faced by a testing procedure: under $\calH_0'$ we observe a subgraph over $\s{Q}$ with $\s{Bern}(q)$ independent edges, while under $\calH_1'$ we observe a subgraph over $\s{Q}$ edges, where there exists a set of $\mathscr{C}$ edges that belong to the planted densest subgraph over $\calK$ (i.e., $\s{Bern}(p)$ random edges), and the remaining edges are independent $\s{Bern}(q)$ random variables. We denote the respective null and alternative distributions by $\pr_{\calH_0'}$ and $\pr_{\calH_1'}$, respectively. We have the following result. 
\begin{lemma}[Total number of planted edge queries]\label{lem:boundTotal}
Assume that $\s{Q}$ satisfies the condition in the first item of Theorem~\ref{thm:1}, and fix $\delta\in(0,1)$. Let $\mathscr{Q}_n$ be any algorithm that makes at most $\s{Q}$ non-adaptive edge queries. Then, with probability at least $1-\delta$,
    \begin{align}
       \mathscr{C}(\mathscr{Q}_n)\leq \s{Q}\frac{k^2}{n^2}\p{1+\frac{1}{\sqrt{\delta}}\frac{n}{k\sqrt{Q}}}.\label{eqn:LemmaProbLower}
    \end{align}
\end{lemma}

\begin{rmk} Note that on the boundary $\s{Q} = \s{Q}^\star$, we have $\frac{n}{k\sqrt{Q}}=o(1)$, as $n\to\infty$. Therefore, \eqref{eqn:LemmaProbLower} reduce to $\mathscr{C}(\mathscr{Q}_n)\leq \s{Q}\frac{k^2}{n^2}\p{1+o(1)}$.

\end{rmk}

\begin{proof}[Proof of Lemma~\ref{lem:boundTotal}]

We will start by proving Lemma~\ref{lem:boundTotal} for deterministic query mechanisms, and then, we address the more general case of randomized algorithms. Since queries are made upfront, they are statistically independent of $\s{G}_n$.
Accordingly, let $\s{X}_e$, for $e\in\mathbb{Q}$, denote an indicator random variable such that $\s{X}_e = 1$ if $e\in\calK\times\calK$, and zero otherwise. Then,
\begin{align}
\mathscr{C}(\mathscr{Q}_n) = \sum_{e\in\mathbb{Q}}\s{X}_e.  
\end{align}
It is clear that $\pr(\mathbb{Q}_\ell\in\calK\times\calK) = \frac{\binom{k}{2}}{\binom{n}{2}} = \frac{k(k-1)}{n(n-1)}\leq\frac{k^2}{n^2}$, and thus,
\begin{align}
\bE \mathscr{C}(\mathscr{Q}_n) = \frac{\binom{k}{2}}{\binom{n}{2}}\s{Q} = \frac{k(k-1)}{n(n-1)}\s{Q}\triangleq \bar{\s{L}},
\end{align}
and furthermore,
\begin{align}
\bE[\mathscr{C}(\mathscr{Q}_n)]^2 &= \bE\sum_{e,e'\in\mathbb{Q}}\s{X}_e\s{X}_{e'}  \\
&= \bE\sum_{e\in\mathbb{Q}}\s{X}_e+\bE\sum_{e\neq e'\in\mathbb{Q}}\s{X}_e\s{X}_{e'}\\
& = \frac{k(k-1)}{n(n-1)}\s{Q}+\frac{k(k-1)}{n(n-1)}\sum_{e\neq e'\in\mathbb{Q}}\pr\pp{e'\in\calK\times\calK\vert e\in\calK\times\calK},\label{eqn:conditionalProb}
\end{align}
where the second equality follows the fact that $\s{X}_e^2 = \s{X}_e$, and in the last equality we note that $\bE\s{X}_e=\pr(e\in\calK\times\calK) = \frac{\binom{k}{2}}{\binom{n}{2}}$. Fix a pair $e = (i^e_1,j^e_1)$, and consider a pair $e = (i^{e'}_1,j^{e'}_1)$. We decompose the summation in \eqref{eqn:conditionalProb} into two parts. In the first part, the edges $e$ and $e'$ are completely disjoint, namely, the do not share a common vertex. We denote this by $e\perp e'$. In this case, $\pr\pp{e'\in\calK\times\calK\vert e\in\calK\times\calK} = \frac{\binom{k-2}{2}}{\binom{n-2}{2}} = \frac{(k-2)(k-3)}{(n-2)(n-3)}$. On the other hand, if $e$ and $e'$ share exactly one common vertex, then, $\pr\pp{e'\in\calK\times\calK\vert e\in\calK\times\calK} = \frac{\binom{k-1}{2}}{\binom{n-1}{2}} = \frac{(k-1)(k-2)}{(n-1)(n-2)}$. It is not difficult to show that $\frac{(k-2)(k-3)}{(n-2)(n-3)}\leq \frac{k(k-1)}{n(n-1)}$ and $\frac{(k-1)(k-2)}{(n-1)(n-2)}\leq \frac{k(k-1)}{n(n-1)}$, for $k<n$, and $n=\omega(1)$. Therefore, 
\begin{align}
\sum_{e\neq e'\in\mathbb{Q}}\pr\pp{e'\in\calK\times\calK\vert e\in\calK\times\calK}\leq \frac{k^2(k-1)^2}{n^2(n-1)^2}\s{Q}^2,
\end{align}
and thus,
\begin{align}
\bE[\mathscr{C}(\mathscr{Q}_n)]^2 &\leq  \frac{k(k-1)}{n(n-1)}\s{Q}+ \frac{k^2(k-1)^2}{n^2(n-1)^2}\s{Q}^2.\label{eqn:conditionalProb2}
\end{align}
Using the above we finally get that,
\begin{align}
\s{Var}\p{\mathscr{C}(\mathscr{Q}_n}&\leq \frac{k(k-1)}{n(n-1)}\s{Q}+ \frac{k^2(k-1)^2}{n^2(n-1)^2}\s{Q}^2-\frac{k^2(k-1)^2}{n^2(n-1)^2}\s{Q}^2\\
& = \frac{k(k-1)}{n(n-1)}\s{Q} = \bar{\s{L}}.
\end{align}
Chebyshev's inequality implies that, for any $\epsilon>0$,
\begin{align}
\pr[\mathscr{C}(\mathscr{Q}_n)\geq(1+\epsilon)\bar{\s{L}}]\leq\frac{1}{\epsilon^2\bar{\s{L}}}.\label{eqn:LbarUpper}
\end{align}
Thus, with probability at least $1-\delta$,
\begin{align}
    \mathscr{C}(\mathscr{Q}_n)\leq \frac{k^2}{n^2}\p{1+\frac{1}{\sqrt{\delta}}\frac{n}{k\sqrt{Q}}}.
\end{align}
Finally, for randomized query mechanisms, we can condition first on the additional source of randomness $\calR$, and apply our arguments above for deterministic query mechanisms, to prove \eqref{eqn:LemmaProbLower}, independently of the realization of $\calR$. Then, by taking an expectation over $\calR$ the proof is concluded.
\end{proof}

\subsection{Hypothesis Test Over the Queried Subgraph}

Provided with the $\s{Q}$ edge queries probed by an arbitrary non-adaptive query mechanism $\mathscr{Q}_n$, we now describe the hypothesis testing problem the detection algorithm $\calA_n$ is faced with, and derive a statistical lower bound on its performance. Recall that the overall distinguishing algorithm is a decomposition $\calA_n\circ\mathscr{Q}_n$. 
While we do not specify the distribution of $\mathscr{C}$, we derived in the previous subsection a high probability upper bound on it. Below, we denote the data that is being supplied to the detection algorithm by a vector $\mathbf{X}\in\{0,1\}^{\s{Q}}$ of length $\s{Q}$, with each entry representing a queried edge. While it makes more sense to represent this data using a matrix (which in turn is a sub-matrix of the original adjacency matrix), we find it more convenient to work with the above vector notation. 

With some abuse of notation, let $(\calH_0',\calH_1')$ denote the hypothesis testing problem faced by $\calA_n$, followed by the \emph{best} edge query mechanism. Specifically, let $\calP \triangleq \s{Bern}(p)$ and $\calQ \triangleq \s{Bern}(q)$. Under $\calH_0'$ it is clear that $\mathbf{X}\sim\pr_{\calH_0'}$ where $\pr_{\calH_0'}\triangleq \calQ^{\otimes \s{Q}}$ is the distribution of a product of $\s{Q}$ Bernoulli random variables $\s{Bern}(q)$. Under $\calH_1'$, the situation is a bit more complicated; we let $\mathbf{X}\sim\pr_{\calH_1'}$, and the distribution $\pr_{\calH_1'}$ is defined as follows. The query trajectory $\mathbb{Q}$ completely determines how the query mechanism operates \emph{for any} realization of $\calK$. Accordingly, conditioned on $\calK$, we let $\calW_{\calK\times\calK}$ denote the set of edge queries that fall inside the planted set $\calK\times\calK$, associated with the best query mechanism. Then, the alternative distribution is constructed as follows:
\begin{enumerate}
    \item We pick $k$ vertices uniformly at random from $[n]$, and denote the obtained set by $\calK$ (this is the original set of vertices over which the subgraph was planted).
    \item Any two vertices in $\calK$ are connected with probability $p$.
    \item Let $\calW_{\calK\times\calK}\subset\calK\times\calK$ be the set of queried edges that fall inside $\calK\times\calK$, and denote its size by $|\calW_{\calK\times\calK}|=\mathscr{C}$.
    \item The elements of $\mathbf{X}$ are the edges in $\calW_{\calK\times\calK}$, and outside this set the edges/elements are drawn i.i.d. $\s{Bern}(q)$.   
\end{enumerate}
Note that due to the result in the previous section, such a set $\calW_{\calK\times\calK}$ of random size $\mathscr{C}$ exists, and that the above hypothesis testing problem is \emph{simple}. Below, we let $\s{Unif}_{n,k}$ be the uniform measure over sets of size $k$ drawn u.a.r. from $[n]$. 
\allowdisplaybreaks
We would like to derive a lower bound on the above testing problem. Specifically, let $\s{R}(\calA_n)$ denote the risk of $\calA_n$, i.e., the sum of the $\s{Type}$-$\s{I}$ and $\s{Type}$-$\s{II}$ error probabilities associated with the detection algorithm $\calA_n$. We start by using the following standard lower bound on the risk $\s{R}(\calA_n)$ of any algorithm \cite[Theorem 2.2]{Tsybakov}, 
\begin{align}
\s{R}(\calA_n)\geq 1 - \s{TV}(\pr_{\calH_0'},\pr_{\calH_1'}).
\end{align}
Next, recall that in the previous subsection we have proved that $\pr(\mathscr{C}>\s{L}^\star)\leq\delta$, where $\s{L}^\star\triangleq \s{Q}\frac{k^2}{n^2}\p{1+o(1)}$. Throughout the rest of the proof, we assume that we are in the regime where $\s{L}^\star\geq1$ (as it is on the boundary $\s{Q}=\s{Q}^\star$), otherwise, detection is clearly impossible. Indeed, if $\s{L}^\star<1$, then this implies that the query mechanism do not probe any planted edge. Thus,
\begin{align}
\s{TV}(\pr_{\calH_0'},\pr_{\calH_1'}) &= \s{TV}(\pr_{\calH_0'},\bE_{\mathscr{C}}\pr_{\calH_1'\vert\mathscr{C}})\\
&\leq \bE_{\mathscr{C}}\pp{\s{TV}(\pr_{\calH_0'},\pr_{\calH_1'\vert\mathscr{C}})}\\
&\leq \delta+\sum_{\ell\leq \s{L}^\star}\s{TV}(\pr_{\calH_0'},\pr_{\calH_1'\vert\mathscr{C}=\ell})\pr(\mathscr{C}=\ell),
\end{align}
where the first inequality follows from the convexity of $(P,Q)\mapsto\s{TV}(P,Q)$. The above total variation distance can be upper bounded as follows \cite[Lemma 2.7]{Tsybakov} 
\begin{align}
2\s{TV}^2(\pr_{\calH_0'},\pr_{\calH_1'\vert\mathscr{C}=\ell})\leq\chi^2(\pr_{\calH_1'\vert\mathscr{C}=\ell},\pr_{\calH_0'}),\label{eqn:NP}
\end{align}
where $\chi^2({\pr}_{\calH_1'\vert\mathscr{C}=\ell},\pr_{\calH_0'})$ is the chi-square distance between ${\pr}_{\calH_1'\vert\mathscr{C}=\ell}$ and $\pr_{\calH_0'}$. It should be clear that this total variation is maximized at the boundary, namely, when $\ell = \s{L}^\star$. Accordingly, without loss of generality, below we focus on this case and ignore the conditioning on $\mathscr{C}$, namely, we treat ${\pr}_{\calH_1'\vert\mathscr{C}}$ as ${\pr}_{\calH_1'}$ with $\mathscr{C}=\s{L}^\star$. 

Let us evaluate the likelihood function $\frac{\pr_{\calH_1'}}{\pr_{\calH_0'}}$. Since,
\begin{align}
\pr_{\calH_1'} = \bE_{\calK\sim\s{Unif}_{n,k}}\pp{\pr_{\calH_1'|\calK,\calW_{\calK\times\calK}}},
\end{align}
it is clear that,
\begin{align}
\frac{\pr_{\calH_1'}}{\pr_{\calH_0'}}(\mathbf{X}) &= \bE_{\calK\sim\s{Unif}_{n,k}}\pp{\frac{{\pr}_{\calH_1'|\calK,\calW_{\calK\times\calK}}}{\calQ^{\otimes\s{Q}}}(\mathbf{X})}\\
&=\bE_{\calK\sim\s{Unif}_{n,k}}\pp{\prod_{i\in\calW_{\calK\times\calK}}f(\mathbf{X}_i)},
\end{align}
where $f\triangleq\frac{\calP}{\calQ}$. Thus,
\begin{align}
\chi^2(\pr_{\calH_1'},\pr_{\calH_0'})&+1 = \bE_{\pr_{\calH_0'}}\pp{\frac{\pr_{\calH_1'}}{\pr_{\calH_0'}}(\mathbf{X})}^2 \\
& = \bE_{\calK_1,\calK_2\sim\s{Unif}_{n,k}}\bE_{\pr_{\calH_0'}}\pp{\prod_{i\in\calW_{\calK_1\times\calK_1}}f(\mathbf{X}_i)\prod_{i\in\calW_{\calK_2\times\calK_2}}f(\mathbf{X}_i)}\\
& = \bE_{\calK_1,\calK_2\sim\s{Unif}_{n,k}}\bE_{\pr_{\calH_0'}}\pp{\prod_{i\in\calW_{\calK_1\times\calK_1}\cap\calW_{\calK_2\times\calK_2}}f^2(\mathbf{X}_i)\prod_{i\in\calW_{\calK_1\times\calK_1}\Delta\calW_{\calK_2\times\calK_2}}f(\mathbf{X}_i)}\\
& = \bE_{\calK_1,\calK_2\sim\s{Unif}_{n,k}}\pp{\prod_{i\in\calW_{\calK_1\times\calK_1}\cap\calW_{\calK_2\times\calK_2}}\bE_{\pr_{\calH_0'}}f^2(\mathbf{X}_i)\prod_{i\in\calW_{\calK_1\times\calK_1}\Delta\calW_{\calK_2\times\calK_2}}\bE_{\pr_{\calH_0'}}f(\mathbf{X}_i)}\\
& = \bE_{\calK_1,\calK_2\sim\s{Unif}_{n,k}}\pp{\prod_{i\in\calW_{\calK_1\times\calK_1}\cap\calW_{\calK_2\times\calK_2}}\bE_{\pr_{\calH_0'}}f^2(\mathbf{X}_i)}\\
& = \bE_{\calK_1,\calK_2\sim\s{Unif}_{n,k}}\pp{(1+\chi^2(\calP,\calQ))^{|\calW_{\calK_1\times\calK_1}\cap\calW_{\calK_2\times\calK_2}|}},
\end{align}
where the last equality holds since $\bE_{\pr_{\calH_0'}}f(\mathbf{X}_i)=1$ and $1+\chi^2(\calP,\calQ) = \bE_{\pr_{\calH_0'}}f^2(\mathbf{X}_i)$. To conclude
\begin{align}
\chi^2(\pr_{\calH_1'},\pr_{\calH_0'})+1 = \bE_{\calK_1,\calK_2\sim\s{Unif}_{n,k}}\pp{(1+\chi^2(\calP,\calQ))^{|\calW_{\calK_1\times\calK_1}\cap\calW_{\calK_2\times\calK_2}|}}.\label{eqn:chisdec}
\end{align}
Now, we note that if $\calK_1\cap\calK_2=\emptyset$ then, obviously, $|\calW_{\calK_1\times\calK_1}\cap\calW_{\calK_2\times\calK_2}|=0$, and thus we may focus on sets $\calK_1$ and $\calK_2$ such that $|\calK_1\cap\calK_2|>0$. Now, given $\calK_1$ and $\calK_2$, the random variable  $\s{H}_{\calK_1\cap\calK_2}\triangleq|\calW_{\calK_1\times\calK_1}\cap\calW_{\calK_2\times\calK_2}|$ is clearly upper bounded by $\s{L}^\star\wedge\binom{|\calK_1\cap\calK_2|}{2}$. Therefore, 
\begin{align}
\chi^2(\pr_{\calH_1'},\pr_{\calH_0'})+1 \leq \bE_{\calK_1,\calK_2\sim\s{Unif}_{n,k}}\pp{(1+\chi^2(\calP,\calQ))^{\s{L}^\star\wedge\binom{|\calK_1\cap\calK_2|}{2}}}.   
\end{align}
Using \eqref{eqn:NP} we have
\begin{align}
2\cdot\s{TV}(\pr_{\calH_0'},\pr_{\calH_1'})^2&\leq \chi^2(\pr_{\calH_1'},\pr_{\calH_0'})\\
&\leq \bE_{\calK_1,\calK_2\sim\s{Unif}_{n,k}}\pp{(1+\chi^2(\calP,\calQ))^{\s{L}^\star\wedge\binom{|\calK_1\cap\calK_2|}{2}}}-1\\
&\leq \bE_{\calK_1,\calK_2\sim\s{Unif}_{n,k}}\pp{\exp\p{\s{L}^\star\wedge\binom{|\calK_1\cap\calK_2|}{2}\cdot\chi^2(\calP,\calQ)}}-1.
\end{align}
Next, notice that $\s{H}\triangleq|\calK_1\cap\calK_2|\sim\s{Hypergeometric}(n,k,k)$, and as so, with this notation, we get
\begin{align}
\chi^2(\pr_{\calH_1'},\pr_{\calH_0'})\leq \bE_{\calK_1,\calK_2\sim\s{Unif}_{n,k}}\pp{\exp\p{\s{L}^\star\wedge\binom{\s{H}}{2}\cdot\chi^2(\calP,\calQ)}}-1.\label{eqn:chiSfind}
\end{align}
We analyze the two possible cases: $\s{L}^\star>\binom{\s{H}}{2}$ and $\s{L}^\star\leq\binom{\s{H}}{2}$. Furthermore, in the sequel we assume that $k<\sqrt{2n}$, and then discuss the complementary region. Starting with the case where $\s{L}^\star>\binom{\s{H}}{2}$, we have
\begin{align}
&\bE\pp{\exp\p{\s{L}^\star\wedge\binom{\s{H}}{2}\cdot\chi^2(\calP,\calQ)}\Ind\ppp{\s{L}^\star>\binom{\s{H}}{2}}}\nonumber\\
&\quad\quad\quad\quad\quad\quad=\bE\pp{\exp\p{\binom{\s{H}}{2}\cdot\chi^2(\calP,\calQ)}\Ind\ppp{\s{L}^\star>\binom{\s{H}}{2}}}.
\end{align}
To analyze this, we use the following tail bound on Hypergeometric random variable (see, e.g., \cite{arias2014community,WuXx}),
\begin{align}
    \pr\p{\s{H}\geq\s{h}}\leq \exp\p{-k\cdot\s{d_{KL}}\p{\s{h}/k||\rho}},\label{eqn:HyprTailBound}
\end{align}
for any $\s{h}/k\geq \rho$, where $\rho\triangleq k/n$. Using the definition of the KL divergence and the identity $(1-x)\log(1-x)\geq -x$, we get
\begin{align}
\s{d_{KL}}\p{a||b} \geq a\log\frac{a}{b}-a.\label{eqn:KLineq}
\end{align}
Using \eqref{eqn:KLineq}, we note that,
\begin{align}
k\cdot\s{d_{KL}}\p{\s{h}/k||\rho} &\geq \s{h}\log\frac{\s{h}}{k\rho}-h\\
& = \s{h}\log\frac{n\s{h}}{k^2}-h.
\end{align}
Accordingly, we have
\begin{align}
&\bE\pp{\exp\p{\binom{\s{H}}{2}\cdot\chi^2(\calP,\calQ)}\Ind\ppp{\s{L}^\star>\binom{\s{H}}{2}}} \nonumber\\
&\quad\quad\quad\quad\quad\quad\quad\quad\quad\leq \pr(\s{H}\leq 1)+\sum_{\s{h}=2}^{\sqrt{2\s{L}^\star}}e^{\frac{\s{h}(\s{h}-1)}{2}\cdot\chi^2(\calP,\calQ)-k\cdot\s{d_{KL}}\p{\s{h}/k||\rho}}\\
&\quad\quad\quad\quad\quad\quad\quad\quad\quad\leq 1+\sum_{\s{h}=2}^{\sqrt{2\s{L}^\star}}e^{\s{h}\p{\frac{\s{h}-1}{2}\cdot\chi^2(\calP,\calQ)-\log\frac{n\s{h}}{k^2}+1}}.\label{eqn:lowerTail0}
\end{align}
For $a>0$ fixed, the function $x\to ax-\log x$ is decreasing on $(0, 1/a)$ and increasing on $(1/a,\infty)$. Therefore,
\begin{align}
\frac{\s{h}-1}{2}\cdot\chi^2(\calP,\calQ)-\log\frac{n\s{h}}{k^2}\leq-\omega,
\end{align}
where
\begin{align}
\omega\triangleq \min\p{\log\frac{n}{k^2}-\frac{1}{2}\cdot\chi^2(\calP,\calQ),\log\frac{n\sqrt{2\s{L}^\star}}{k^2}-\frac{\sqrt{2\s{L}^\star}-1}{2}\cdot\chi^2(\calP,\calQ)}.\label{eqn:omega}
\end{align}
Substituting $\s{L}^\star = k^2\s{Q}/\s{n^2}$, while noticing that $\log\frac{n\sqrt{2\s{L}^\star}}{k^2} = \log\frac{n}{k}+O(\log\frac{\log n}{k})= [1-o(1)]\cdot\log\frac{n}{k}$, the second term in the minimum tends to $\infty$ if 
\begin{align}
\s{Q}<(2-\epsilon)\cdot\frac{n^2}{k^2\chi^4(\calP,\calQ)}\log^2\frac{n}{k},\label{eqn:finalCond1}
\end{align}
for any $\epsilon>0$. This is also the case of the first term, since
\begin{align}
\log\frac{n}{k^2}-\frac{1}{2}\cdot\chi^2(\calP,\calQ) = \log\frac{n\sqrt{\s{L}^\star}}{k^2}-\frac{\sqrt{\s{L}^\star}-1}{2}\cdot\chi^2(\calP,\calQ)+\frac{\sqrt{\s{L}^\star}}{2}\chi^2(\calP,\calQ)-\log\sqrt{\s{L}^\star},
\end{align}
with the second difference bounded from below if $\chi^2(\calP,\calQ)$ is finite. Hence the sum in \eqref{eqn:lowerTail0} converges to zero, and we get
\begin{align}
\bE\pp{\exp\p{\binom{\s{H}}{2}\cdot\chi^2(\calP,\calQ)}\Ind\ppp{\s{L}^\star>\binom{\s{H}}{2}}}\leq1+o(1).\label{eqn:lowerTail}
\end{align}
Next, we turn to the case where $\s{L}^\star\leq\binom{\s{H}}{2}$. We have
\begin{align}
&\bE\pp{\exp\p{\s{L}^\star\wedge\binom{\s{H}}{2}\cdot\chi^2(\calP,\calQ)}\Ind\ppp{\s{L}^\star
\leq\binom{\s{H}}{2}}}\nonumber\\
&\quad\quad\quad\quad\quad\quad=\exp\p{\s{L}^\star\cdot\chi^2(\calP,\calQ)}\cdot\pr\pp{\s{L}^\star\leq\binom{\s{H}}{2}}\\
&\quad\quad\quad\quad\quad\quad\leq\exp\p{\s{L}^\star\cdot\chi^2(\calP,\calQ)-\sqrt{2\s{L}^\star}\log\frac{n\sqrt{2\s{L}^\star}}{k^2}+\sqrt{2\s{L}^\star}}\\
&\quad\quad\quad\quad\quad\quad=\exp\pp{-\sqrt{2\s{L}^\star}\p{\log\frac{n\sqrt{2\s{L}^\star}}{k^2}-1-\sqrt{\frac{\s{L}^\star}{2}}\cdot\chi^2(\calP,\calQ)}},\label{eqn:expo}
\end{align}
where the inequality follows from \eqref{eqn:HyprTailBound} and \eqref{eqn:KLineq}. Substituting $\s{L}^\star=k^2\s{Q}/n^2$, and noticing that $\log\frac{n\sqrt{2\s{L}^\star}}{k^2} = \log\frac{n}{k}+O(\log\frac{\log n}{k})= [1-o(1)]\cdot\log\frac{n}{k}$, it is clear that the right hand side of \eqref{eqn:expo} converges to zero as long as
\begin{align}
\s{Q}<(2-\epsilon)\cdot\frac{n^2}{k^2\chi^4(\calP,\calQ)}\log^2\frac{n}{k},\label{eqn:finalCond2}
\end{align}
for any $\epsilon>0$. Accordingly, under this condition \eqref{eqn:expo} converges to zero. Combining \eqref{eqn:chiSfind}, \eqref{eqn:lowerTail}, and \eqref{eqn:expo}, we get that for $k<\sqrt{2n}$,
\begin{align}
    \chi^2(\pr_{\calH_1'},\pr_{\calH_0'})\leq 1+o(1)-1 = o(1),
\end{align}
provided that \eqref{eqn:finalCond1} and \eqref{eqn:finalCond2} hold (which coincide), as required (see, \eqref{eqn:ITlimit} in Theorem~\ref{thm:1}). Finally, we mention that the complementary case, where $k>\sqrt{2n}$ follow from almost the same arguments as above, with the following small modifications; specifically, in this regime, we use the following tail bound on Hypergeometric random variable
\begin{align}
    \pr\p{\s{H}<\s{h}}\leq \exp\p{-k\cdot\s{d_{KL}}\p{1-\s{h}/k||1-\rho}},\label{eqn:HyprTailBound2}
\end{align}
for any $\s{h}/k<\rho$, where $\rho=k/n$, and we lower bound the KL divergence using,
\begin{align}
\s{d_{KL}}\p{a||b} \geq -a+(1-a)\log\frac{1-a}{1-b}.\label{eqn:KLineq0}
\end{align}
Then, repeating the same steps above it can be shown that now the lower tail term is asymptotically small, namely,
\begin{align}
\bE\pp{\exp\p{\binom{\s{H}}{2}\cdot\chi^2(\calP,\calQ)}\Ind\ppp{\s{L}^\star>\binom{\s{H}}{2}}}\leq o(1),
\end{align}
while the upper tail term is,
\begin{align}
\bE\pp{\exp\p{\s{L}^\star\wedge\binom{\s{H}}{2}\cdot\chi^2(\calP,\calQ)}\Ind\ppp{\s{L}^\star
\leq\binom{\s{H}}{2}}}\leq1+o(1),
\end{align}
and thus $\chi^2(\pr_{\calH_1'},\pr_{\calH_0'})\leq 1+o(1)-1 = o(1)$, 
provided that the same conditions as in \eqref{eqn:finalCond1} and \eqref{eqn:finalCond2} hold. This concludes the proof (see, \eqref{eqn:ITlimit} in Theorem~\ref{thm:1}). 

\section{Conclusion and Outlook}\label{sec:conc}

In this paper, we formulated and analyzed a variant of the classical $\s{PDS}$ problem, where one can only observe a small part of the graph using non-adaptive edge queries. This problem is relevant, for example, when access to the edges (connections) between vertices (individuals) may be scarce due to privacy concerns. For this model, we derived the number of queries necessary and sufficient for detecting the presence of the planted subgraph, up to a constant factor. For the special case of planted cliques, our results are completely tight. 
This work also has left number of specific problems open, including the following:
\begin{itemize}
   \item It would be quite interesting to provide any concrete evidence for our conjectured statistical-computational gap either by    means of an average-case reduction from the planted clique problem, or failure of classes of powerful algorithms (such as, sum-of-squares hierarchy, low-degree polynomials, etc.), below the computational barrier.
    \item Our bounds are almost tight in the sense that there is a multiplicative constant gap between our lower and upper bounds. As was mentioned in the paper, we believe that the source for this gap is our lower bound; the $\chi^4(p||q)$ factor should be in fact $d_{\s{KL}}^2(p||q)$. We suspect that one way to prove this is by applying a truncation procedure on the likelihood analysis when deriving the lower bound on the risk.  
    \item In this paper, we analyzed the regime where the edge probabilities $p$ and $q$ are fixed and independent of $n$. The regime where $p$ and $q$ depend on $n$, e.g., both decay polynomially in $n$, is quite important and challenging. It would be interesting to find the phase diagram for this case. Note that here $\chi^2(p||q)$ is not a constant anymore, but decays with $n$ polynomially fast.
    \item While in this paper we have focused on the detection problem, it is interesting to consider the recovery and partial recovery variants of our setting as well. 
\end{itemize}


\bibliographystyle{alpha}
\bibliography{bibfile}

\newcommand{\etalchar}[1]{$^{#1}$}
\begin{thebibliography}{AHHM20}

\bibitem[ACV14]{arias2014community}
Ery Arias-Castro and Nicolas Verzelen.
\newblock {Community detection in dense random networks}.
\newblock {\em The Annals of Statistics}, 42(3):940 -- 969, 2014.

\bibitem[AHHM20]{Alweiss2020}
Ryan Alweiss, Chady~Ben Hamida, Xiaoyu He, and Alexander Moreira.
\newblock On the subgraph query problem.
\newblock {\em Combinatorics, Probability and Computing}, page 1–16, Jul
  2020.

\bibitem[ALL{\etalchar{+}}16]{Anagnostopoulos2016CommunityDO}
A.~Anagnostopoulos, Jakub Lacki, Silvio Lattanzi, S.~Leonardi, and Mohammad
  Mahdian.
\newblock Community detection on evolving graphs.
\newblock In {\em NIPS}, 2016.

\bibitem[BB20]{brennan20a}
Matthew Brennan and Guy Bresler.
\newblock Reducibility and statistical-computational gaps from secret leakage.
\newblock In {\em Proceedings of Thirty Third Conference on Learning Theory},
  volume 125, pages 648--847, 09--12 Jul 2020.

\bibitem[BBH18]{brennan2018reducibility}
Matthew Brennan, Guy Bresler, and Wasim Huleihel.
\newblock Reducibility and computational lower bounds for problems with planted
  sparse structure.
\newblock In {\em COLT}, pages 48--166, 2018.

\bibitem[BBH19]{Brennan2019UniversalityOC}
Matthew Brennan, Guy Bresler, and Wasim Huleihel.
\newblock Universality of computational lower bounds for submatrix detection.
\newblock In {\em COLT}, 2019.

\bibitem[BI13]{butucea2013detection}
Cristina Butucea and Yuri~I Ingster.
\newblock Detection of a sparse submatrix of a high-dimensional noisy matrix.
\newblock {\em Bernoulli}, 19(5B):2652--2688, 2013.

\bibitem[BR13]{berthet2013complexity}
Quentin Berthet and Philippe Rigollet.
\newblock Complexity theoretic lower bounds for sparse principal component
  detection.
\newblock In {\em Proceedings of the 26th Annual Conference on Learning
  Theory}, volume~30, pages 1046--1066, 12--14 Jun 2013.

\bibitem[CC18]{candogan2018finding}
Utkan~Onur Candogan and Venkat Chandrasekaran.
\newblock Finding planted subgraphs with few eigenvalues using the schur--horn
  relaxation.
\newblock {\em SIAM Journal on Optimization}, 28(1):735--759, 2018.

\bibitem[CFGH18]{Conlon18}
D.~Conlon, J.~Fox, A.~Grinshpun, and X.~He.
\newblock Online ramsey numbers and the subgraph query problem.
\newblock {\em arXiv: Combinatorics}, pages 159--194, 2018.

\bibitem[CLR17]{cai2015computational}
Tony Cai, Tengyuan Liang, and Alexander Rakhlin.
\newblock Computational and statistical boundaries for submatrix localization
  in a large noisy matrix.
\newblock {\em Annals of Statistics}, 45(4):1403--1430, 08 2017.

\bibitem[CX16]{chen2016statistical}
Yudong Chen and Jiaming Xu.
\newblock Statistical-computational tradeoffs in planted problems and submatrix
  localization with a growing number of clusters and submatrices.
\newblock {\em Journal of Machine Learning Research}, 17(27):1--57, 2016.

\bibitem[FGN{\etalchar{+}}20]{Feige20}
Uriel Feige, David Gamarnik, Joe Neeman, Miklós~Z. Rácz, and Prasad Tetali.
\newblock Finding cliques using few probes.
\newblock {\em Random Structures \& Algorithms}, 56(1):142--153, 2020.

\bibitem[FKSV15]{Ferber15}
Asaf Ferber, Michael Krivelevich, Benny Sudakov, and Pedro Vieira.
\newblock Finding hamilton cycles in random graphs with few queries.
\newblock {\em Random Struct. Algorithms}, 49, 05 2015.

\bibitem[FKSV17]{Ferber17}
Asaf Ferber, M.~Krivelevich, B.~Sudakov, and Pedro Vieira.
\newblock Finding paths in sparse random graphs requires many queries.
\newblock {\em Random Struct. Algorithms}, 50:71--85, 2017.

\bibitem[GMZ17]{gao2017sparse}
Chao Gao, Zongming Ma, and Harrison~H Zhou.
\newblock Sparse {CCA}: Adaptive estimation and computational barriers.
\newblock {\em The Annals of Statistics}, 45(5):2074--2101, 2017.

\bibitem[HKW16]{Hartmann2016ClusteringEN}
Tanja Hartmann, A.~Kappes, and D.~Wagner.
\newblock Clustering evolving networks.
\newblock In {\em Algorithm Engineering}, 2016.

\bibitem[HWX15]{hajek2015computational}
Bruce~E Hajek, Yihong Wu, and Jiaming Xu.
\newblock Computational lower bounds for community detection on random graphs.
\newblock In {\em COLT}, pages 899--928, 2015.

\bibitem[HWX16]{hajek2016achieving}
Bruce Hajek, Yihong Wu, and Jiaming Xu.
\newblock Achieving exact cluster recovery threshold via semidefinite
  programming.
\newblock {\em IEEE Transactions on Information Theory}, 62(5):2788--2797,
  2016.

\bibitem[HWX17]{Hajek17}
Bruce Hajek, Yihong Wu, and Jiaming Xu.
\newblock Information limits for recovering a hidden community.
\newblock {\em IEEE Transactions on Information Theory}, 63(8):4729--4745,
  2017.

\bibitem[MAC20]{mardia2020finding}
Jay Mardia, Hilal Asi, and Kabir~Aladin Chandrasekher.
\newblock Finding planted cliques in sublinear time.
\newblock {\em ArXiv}, abs/2004.12002, 2020.

\bibitem[Mon15]{montanari2015finding}
Andrea Montanari.
\newblock Finding one community in a sparse graph.
\newblock {\em Journal of Statistical Physics}, 161(2):273--299, 2015.

\bibitem[MS17a]{Mazumdar2017ClusteringWN}
A.~Mazumdar and B.~Saha.
\newblock Clustering with noisy queries.
\newblock In {\em NIPS}, 2017.

\bibitem[MS17b]{Mazumdar2017QueryCO}
A.~Mazumdar and B.~Saha.
\newblock Query complexity of clustering with side information.
\newblock In {\em NIPS}, 2017.

\bibitem[MW15]{ma2015computational}
Zongming Ma and Yihong Wu.
\newblock Computational barriers in minimax submatrix detection.
\newblock {\em The Annals of Statistics}, 43(3):1089--1116, 2015.

\bibitem[RS20]{racz2020finding}
Miklós~Z. Rácz and Benjamin Schiffer.
\newblock Finding a planted clique by adaptive probing.
\newblock {\em ALEA Latin American Journal of Probability and Mathematical
  Statistics}, 17:775–790, 2020.

\bibitem[Tsy08]{Tsybakov}
Alexandre~B. Tsybakov.
\newblock {\em Introduction to Nonparametric Estimation}.
\newblock Springer Publishing Company, Incorporated, 1st edition, 2008.

\bibitem[VAC{\etalchar{+}}15]{verzelen2015community}
Nicolas Verzelen, Ery Arias-Castro, et~al.
\newblock Community detection in sparse random networks.
\newblock {\em The Annals of Applied Probability}, 25(6):3465--3510, 2015.

\bibitem[VH16]{Vinayak2016CrowdsourcedCQ}
Ramya~Korlakai Vinayak and B.~Hassibi.
\newblock Crowdsourced clustering: Querying edges vs triangles.
\newblock In {\em NIPS}, 2016.

\bibitem[WBP16]{wang2016average}
Tengyao Wang, Quentin Berthet, and Yaniv Plan.
\newblock Average-case hardness of rip certification.
\newblock In {\em Advances in Neural Information Processing Systems}, pages
  3819--3827, 2016.

\bibitem[WBS16]{wang2016statistical}
Tengyao Wang, Quentin Berthet, and Richard~J Samworth.
\newblock Statistical and computational trade-offs in estimation of sparse
  principal components.
\newblock {\em The Annals of Statistics}, 44(5):1896--1930, 2016.

\bibitem[WX20]{wu2018statistical}
Yihong Wu and Jiaming Xu.
\newblock Statistical problems with planted structures: Information-theoretical
  and computational limits.
\newblock In Miguel R.~D. Rodrigues and Yonina~C. Eldar, editors, {\em
  Information-Theoretic Methods in Data Science}. Cambridge University Press,
  Cambridge, 2020.

\bibitem[WX23]{WuXx}
Yihong Wu and Jiaming Xu.
\newblock {\em Statistical inference on graphs: Selected Topics}.
\newblock Lecture notes, 2023.

\end{thebibliography}
\end{document}